\documentclass[11pt, letterpaper]{article}
\usepackage[margin=1in]{geometry} 
\usepackage{setspace} 
\usepackage[colorlinks=true, linkcolor=blue, citecolor=blue, urlcolor=blue, backref]{hyperref}

\usepackage{fullpage}
\usepackage{amsmath}
\usepackage{amsthm}
\usepackage{amssymb}
\usepackage{amsfonts}
\usepackage{hyperref}
\usepackage{cleveref} 
\usepackage{color}
\usepackage{xcolor}
\usepackage{mathbbol}
\usepackage{multicol}
\usepackage{relsize}
\usepackage{algorithm}
\usepackage{algorithmic}
\usepackage{enumitem}
\newlist{customitemize}{itemize}{3}
\setlist[customitemize,1]{label=RC\arabic{customitemizei}.}

\newcommand{\ignore}[1]{}

\newtheorem{theorem}{Theorem}

\newtheorem{lemma}[theorem]{Lemma}

\renewcommand{\Pr}{{\bf Pr}}

\newcommand{\ZZ}{\mathbb{Z}}

\newcommand{\fa}{\mathfrak{a}}

\begin{document}

\title{Sublinear Time Algorithms for Abelian Group\\ Isomorphism and Basis Construction}
\author{{\bf Nader H. Bshouty}\\ Dept. of Computer Science\\ Technion,  Haifa\\ }

\maketitle
\begin{abstract}
\color{black}
In this paper, we study the problems of Abelian group isomorphism and basis construction in two models. In the {\it partially specified model} (PS-model), the algorithm does not know the group size but can access randomly chosen elements of the group along with the Cayley table of those elements, which provides the result of the binary operation for every pair of selected elements. In the stronger {\it fully specified model} (FS-model), the algorithm knows the size of the group and has access to its elements and Cayley table. 

Given two Abelian groups, $G$, and $H$, we present an algorithm in the PS-model (and hence in the FS-model) that runs in time $\tilde O(\sqrt{|G|})$ and decides if they are isomorphic. This improves on Kavitha's linear-time algorithm and gives the first sublinear-time solution for this problem. We then prove the lower bound $\Omega(|G|^{1/4})$ for the FS-model and the tight bound $\Omega(\sqrt{|G|})$ for the PS-model. This is the first known lower bound for this problem. We obtain similar results for finding a basis for Abelian groups. For deterministic algorithms, a simple $\Omega(|G|)$ lower bound is given. 
\end{abstract}
\section{Introduction}

In this paper, we study the problems of Abelian group isomorphism and constructing a basis in two models. In the {\it fully specified model}\footnote{Also known as the Cayley or multiplication table model} (FS-model), the size of the groups is known, and the algorithm has access to the elements of the group and their Cayley tables. In the {\it partially specified model} (PS-model), the size of the group is unknown; the algorithm can receive uniform random elements of the groups and access the Cayley table of elements observed so far.

To the best of our knowledge, we provide the first sublinear-time algorithm (in the size of the group) for isomorphism testing and basis construction problems of Abelian groups in these models while also establishing the first tight lower bounds.

All results in this paper are stated for the PS-model, unless explicitly indicated otherwise for the FS-model. Additionally, we assume that the computation is performed in the RAM model, where reading each element of the group requires one unit of time. 

\subsection{Generators For Abelian Groups}
To address these problems, we first study the problem of constructing a set of generators for Abelian groups $G$ using Oracle access to the Cayley table of $G$. Our algorithm is randomized and runs in time $\tilde O(\sqrt{|G|})$ and constructs a set of generators $A=\{a_i\}_{i\in [t]}$ for $G$ with size at most $t\le \log |G|$, satisfying the {\it triangular relations}\footnote{A {\it Relation} is an equation for elements of the group. For example, $a^2b=c^3d$ for some elements $a,b,c,d$ in a group. For a set of generators $\{a_1,a_2,\ldots,a_t\}$, triangular relations are relations of the form $a_i^{k_i}=a_1^{\lambda_{i,1}}\cdots a_{i-1}^{\lambda_{i,i-1}}$ for all $i\in [t]$}. 
We then use these relations to construct an Abelian group $\Gamma_G$ isomorphic to~$G$, that is a subset of some quotient ring, where each operation in $\Gamma_G$ can be performed in time $poly(\log |G|)$. 

Our algorithm operates within the PS-model. We also show that this algorithm is optimal, i.e., any randomized algorithm in the PS-model that produces a set of generators with any set of relations of size at most $|G|^{o(1)}$ must run in time $\Omega(\sqrt{|G|})$. For the FS-model, we obtain the lower bound $\Omega(|G|^{1/4})$.

We then apply this to isomorphism testing and basis construction. 
\color{black}

\color{black}
\subsection{Isomorphism Testing}
Group isomorphism is a fundamental problem in group theory and computation  \cite{Babai16,BabaiCGQ11,BabaiQ12,BabaiS84,ChenGQTZ24,DietrichW21,GarzonZ91,GrochowQ17,GrochowQ24,Kavitha07,KaragiorgosP11,KaragiorgosP112,Miller78,LiptonSZ76,Rosen,Savage80,Vikas96}. 

In this paper, we focus on the Abelian group isomorphism problem~\cite{GarzonZ91,KaragiorgosP11,KaragiorgosP112,Kavitha07,Miller78,LiptonSZ76,Savage80,Vikas96}: Given two finite Abelian groups $H$ and $G$ with access to their Cayley tables. Decide if $H$ is isomorphic to~$G$. 

Lipton et al. ~\cite{LiptonSZ76} 
showed that Abelian group isomorphism could be solved in polynomial time. Savage~\cite{Savage80} gave an algorithm that runs in time $O(|G|^2)$. Vikas~\cite{Vikas96} improved this bound and gave an $O(|G|\log |G|)$ time algorithm. Then Kavitha~\cite{Kavitha07} and Karagiorgos and Paulakis \cite{KaragiorgosP11,KaragiorgosP112} improved it to $O(|G|)$. All the algorithms listed above are deterministic. In this paper, we show.
\begin{theorem}
    Any deterministic algorithm that decides whether an Abelian group is isomorphic to a given group $G$ must make at least $\Omega({|G|})$ accesses to the elements of the groups and to the Cayley table of $G$. 
\end{theorem}

Therefore, for deterministic algorithms, the algorithms of Kavitha~\cite{Kavitha07} and Karagiorgos and Paulakis  \cite{KaragiorgosP11,KaragiorgosP112} are optimal.

In this paper, we give a randomized algorithm for this problem. We show.
\begin{theorem}
    The Abelian group isomorphism problem can be solved in time\footnote{All results in this paper are stated for the PS-model, unless explicitly indicated otherwise for the FS-model.} $\tilde O(\sqrt{|G|})$. The algorithm also computes an explicit isomorphism. 
\end{theorem}

We also establish tight lower bounds for this problem.
\begin{theorem}
    Any Abelian group isomorphism algorithm must make at least $\Omega(\sqrt{|G|})$ accesses to the elements of the group and the Cayley table of $G$. 
\end{theorem}
To the best of our knowledge, this is the first sublinear upper bound and the first tight lower bound for this problem.

In the FS-model, we establish the following lower bound.
\begin{theorem}
    In the FS-model, any Abelian group isomorphism algorithm must make at least $\Omega({|G|^{1/4}})$ access to the elements of the groups and the Cayley table of $G$. 
\end{theorem}

We also note that the problem of determining the existence of a homomorphism between two Abelian groups $G$ and $H$ can also be solved with the same time complexity. Our algorithm also finds the homomorphism and isomorphism. 

\subsection{Basis Construction}
Computing a basis for a finite Abelian group is a fundamental problem in computational group theory, with applications in cryptography, coding theory, and algebraic computations~\cite{Borges-QuintanaBM05,BuchmannS05,ChenF11,GarzonZ91,Iliopoulos85,Iliopoulos89,KaragiorgosP11,KaragiorgosP112,KoblitzM04,Teske99}. 

Every finite Abelian group $G$ can be represented as a direct product of cyclic groups $G_1\times G_2\times \cdots \times G_t$. If $a_i$ generates $G_i$, then the elements $a_1,a_2,\ldots,a_t$ are called {\it basis} for $G$.

All known sublinear time algorithms that find a basis for an Abelian group run in time $\tilde O(\sqrt{|G|})$ and assume that the group is given by a set of generators along with their orders and the Cayley table~\cite{BuchmannS05,ChenF11,Iliopoulos85,KaragiorgosP11,Teske99}. Chen and Fu \cite{ChenF11}, and Karagiorgos and Paulakis \cite{KaragiorgosP11,KaragiorgosP112} gave an $O(|G|)$ time deterministic algorithm that accesses only the Cayley table. We show that this algorithm is optimal.
\begin{theorem}
    Any deterministic algorithm that finds a basis for an Abelian group $G$ must access the elements of the group and its Cayley table at least $\Omega({|G|})$ times.
\end{theorem}

In this paper, we give a randomized algorithm for this problem. We show.
\begin{theorem}
    There exists a randomized algorithm that accesses the Cayley table of an Abelian group $G$ at most $\tilde O(\sqrt{|G|})$ time and finds a basis for $G$.
\end{theorem}
To the best of our knowledge, this result provides the first sublinear algorithm for finding the basis of an Abelian group using only the Cayley table. 

We then establish the following tight lower bound.
\begin{theorem}
    Any randomized algorithm that finds a basis for an Abelian group $G$ must access the elements of the group and its Cayley table at least $\Omega(\sqrt{|G|})$ times.
\end{theorem}
This is the first tight lower bound for the problem.

Furthermore, in the FS-model, we show.

\begin{theorem}
    Any randomized algorithm that finds a basis for an Abelian group $G$ in the FS-model must access the elements of the group and its Cayley table at least $\Omega({|G|^{1/4}})$ times.
\end{theorem}

\color{black}

\section{Our Technique}
In this section, we outline the techniques used in our algorithms, beginning with those applied to establish the upper bounds, followed by the methods for deriving the lower bounds.

\subsection{Generators with Triangular Relation}\label{GwTR}
We first provide an algorithm that, for every Abelian group $G$, runs in time $\tilde O(\sqrt{|G|})$ and constructs a set of generators $A=\{a_1,\ldots,a_t\}$ of size at most $t=\log |G|$ that satisfies {\it triangular relations}. That is, for every $i\in [t]$, there exists a relation of the form $a_i^{k_i}=a_1^{\lambda_{i,1}}\cdots a_{i-1}^{\lambda_{i,i-1}}$, where $k_i\ge 2$.

We present two algorithms. A deterministic algorithm that runs in time $O(|G|)$ and a randomized algorithm that runs in time $\tilde O(\sqrt{|G|})$. In both algorithms, the main idea is to construct a sequence of subgroups $\{e\}< G_1< G_2< \cdots< G_t=G$, where $G_i$ is generated by $\{a_1,\ldots,a_i\}$. To construct $G_{i+1}$, we choose an element $a_{i+1}$ in $G\backslash G_{i}$, find the order $k_{i+1}$ of $a_{i+1}$ in the quotient group $G/G_i$, and construct the elements of all the disjoint cosets of $G_i$ to form $G_{i+1}$ and establish the relation for $a_{i+1}$. 

In the deterministic algorithm, we leverage the fact that elements of different cosets are disjoint, which guarantees that each group element is generated exactly once. As a result, the algorithm enumerates all elements of the group without repetition and runs in time $\tilde{O}(|G|)$.

In the randomized algorithm, we use the following simple idea. Let ${\cal F}$ be a homomorphism ${\cal F}:G\to F$ with a readily computable inverse ${\cal F}^{-1}$. For $a\in G$, to determine ${\cal F}(a)$, we choose a set of $\tilde O(\sqrt{|G|})$ elements $R_1$ in $G$ with known ${\cal F}$ values, then choose another set $R_2$ of $\tilde O(\sqrt{|G|})$ elements in $G$ uniformly at random, also with known ${\cal F}$ values. This can be done by choosing  $f\in F$ uniformly at random, computing $\chi={\cal F}^{-1}(f)$ and then choosing $x\in \chi$ uniformly at random. Since $aG=G$, with high probability, $R_1\cap aR_2$ contains an element $b=ac$, where $b\in R_1$ and $c\in R_2$. Then ${\cal F}(a)={\cal F}(bc^{-1})={\cal F}(b){\cal F}(c)^{-1}$.

We use this technique to find an element $a_{i+1}$ in $G\backslash G_i$, determine its order in the quotient group $G/G_i$, and compute the relation for $a_{i+1}$ in time $\tilde O(\sqrt{|G|})$. 
\color{black}
\color{black}

\subsection{Isomorphism and Basis Construction}
Given an Abelian group $G$, we use the algorithm from the previous section to find a set of generators $A$ with triangular relations in time $\tilde O(\sqrt{|G|})$. Using these relations, ${\cal R}$, we construct a group $\Gamma({\cal R})$ that is isomorphic to $G$, which we call the monomial Abelian group derived from ${\cal R}$. 
This is because the elements of this group are monomials over $t=|A|$ formal elements $x_1,\ldots,x_t$, where multiplication in this group is modulo polynomials corresponds to the relations ${\cal R}$. 
The elements $x_i$, satisfy the same relations as $a_i$ in $G$. 
We also show that the multiplication and inversion operation in $\Gamma({\cal R})$ can be performed in time $poly(t)$. 

Then, one can apply the Smith normal form to convert these generators and relations into a basis. The time complexity of this method is $poly(\log |G|)$~\cite{KannanB79}. 

This addresses both the basis construction problem and the isomorphism problem.

\subsection{Our Technique for the Lower Bounds}
In this section, we present the technique used to establish two foundational lower bounds from which all other lower bounds are derived. Our approach is inspired by the technique introduced in~\cite{GallY13}.

For the lower bound in the FS-model, we first define the class of all Abelian groups, ${\cal G}$ on the ground set $[p^{2m}]$ that are isomorphic to either $H_1={\mathbb{Z}}_{p^2}^m$ or $H_2=\mathbb{Z}_{p^2}^{m-1}\times \mathbb{Z}_p^2$. We then show that any algorithm which, with probability at least $2/3$, determines whether $G\in{\cal G}$ is isomorphic to $H_1$ or $H_2$, must access the elements of $G$ and its Cayley table at least $\Omega(|G|^{1/4})$ times.

We apply Yao's Minimax Principle and show that, for any deterministic algorithm~$A$ that decides whether a given group~$G$ is isomorphic to~$H_1$ or~$H_2$ (by outputting~$1$ or~$2$), if~$A$ is executed on a uniformly random group~$G \in \mathcal{G}$, then with probability at least~$2/3$,~$A$ must access elements of~$G$ and entries of its Cayley table at least~$\Omega(|G|^{1/4})$ times.

To this end, consider such an algorithm ${\cal A}$. We map the elements of the group $G\in {\cal G}$ to abstract elements $G'=\{\sigma_1,\ldots,\sigma_{|G|}\}$ through a bijective function $\phi:G\to G'$ chosen uniformly at random. Under this mapping, $G'$ is isomorphic to $G$ with the product $\sigma_i\sigma_j=\phi(\phi^{-1}(\sigma_i)\phi^{-1}(\sigma_j))$. Next, we map $G'$ to the free $\mathbb{Z}_{p^2}$-module $L=\{\alpha_1x_1+\cdots+\alpha_{w}x_{w}|\alpha_i\in \mathbb{Z}_{p^2},w\in[|G|]\}$ with coefficients in $\mathbb{Z}_{p^2}$, where $\{x_i\}_{i=1}^\infty$ are formal elements. In this mapping, each $\sigma_i$ is mapped to $x_i$. We then run the algorithm ${\cal A}$ on this mapped group, i.e., we run the algorithm and replace each multiplication in the group $G'$ with $+$ in $L$ and each $\sigma_i$ with $x_i$. 

Since in ${\cal A}$ each $\sigma_i$ is replaced with an element in $L$, and comparisons in the IF commands (i.e., the command ``If $x=y$ Then ... Else ...'') involve elements in $L$, the execution follows a single, well-defined path $P$ in the algorithm~${\cal A}$. Assume that this path results in an output of $1$. 

We then show that if the path $P$ contains at most $|G|^{1/4}$ group operations and elements of the group, then, for $\phi$ chosen uniformly at random, running the algorithm on $G'$, with high probability, the logical values of the comparisons in the IF commands in the algorithm ${\cal A}$ are consistent with the path $P$. That is, with high probability, the execution on $G'$ will follow the same path $P$. The intuition here is that both $H_1$ and $H_2$ behave similarly to $L$ under a randomly chosen $\phi$.   
Since this path outputs~$1$, the failure probability of ${\cal A}$ is close to $1/2$. 

In the PS-model, we define the class of all Abelian groups ${\cal G}$ that are isomorphic to either $G_1=\mathbb{Z}_p^m$ or $G_2=\mathbb{Z}_p^{m-1}$. We show that if there exists an algorithm running in time $T$ that, with probability at least $2/3$, distinguishes between these groups, then there is also an algorithm running in time $T$ that, with probability at least $2/3$, distinguishes between $H_1$ and $H_2$. 
This result follows from the fact that $pH_1$ is isomorphic to $G_1$ and $pH_2$ is isomorphic to $G_2$. Consequently, any algorithm for $G_i$ gives an algorithm for $H_i$. 
Hence, we establish that $T=\Omega(|H_1|^{1/4})=\Omega(p^{m/2})=\Omega(\sqrt{|G_1|})$.

For deterministic algorithms, consider the two Abelian groups $D_1=\ZZ_p^m$ and $D_2=\ZZ_p^{m-2}\times \ZZ_{p^2}$. An adversary can provide elements from $W=\ZZ_p^{m-2}\times \{(0,0)\}$ for each access to an element $\sigma_i$. The sum of any two elements in $W$ remains in $W$, preventing the algorithm from distinguishing between the Abelian groups until the $(p^{m-2}+1)$-th element is requested. At this point, the adversary must reveal whether the hidden group is $D_1$ or $D_2$. Thus, the complexity of the algorithm is at least $\Omega(p^{m-2})=\Omega(n)$ for any constant $p$. 

Using these two lower bounds, we get the lower bounds for isomorphism testing and basis construction.

\section{Abelian Groups and Preliminary Results}
\color{black} 
In this section, we provide definitions and preliminary results that will be used to prove the main results.
\subsection{Basic Facts in Groups}
Let $(G,\cdot)$ be a finite Abelian (commutative) group. We denote by $e$ the identity element of $G$. For two sets $A,B\subseteq G$ we define $AB=\{ab|a\in A,b\in B\}$. For a singleton $A=\{a\}$ or $B=\{b\}$, we simply write $aB$ and $Ab$, respectively. 
When $H$ is a subgroup of $G$, we write $H\le G$.

For a nonempty subset $A$ of $G$, 
$\langle A\rangle:=\{a_1\cdots a_k|a_i\in A,k\in\mathbb{N}\}$
is the subgroup {\it generated} by~$A$. If $\langle A\rangle =G$, then we say that $A$ is a {\it set of generators} of $G$. We also write $\langle A_1,A_2,\ldots,A_t\rangle$ for $\left\langle \cup_{i=1}^tA_i\right\rangle$, and replace $A_i$ with $a_i$ when $A_i=\{a_i\}$ is a singleton.

Let $H\le G$ be a subgroup of $G$. For $a\in G$, we call $Ha$ a {\it coset}. The following are known results and can be easily proved.
\begin{lemma}
    Let $H\le G$. Then
    \begin{enumerate}
        \item $|H|$ divides $|G|$.
        \item For every $a\in G$, we have $|aH|=|H|$.
        \item $a\in H$ if and only if $aH=H$.
        \item If $a\not\in H$ then $aH\cap H=\emptyset$.
        \item $aH=bH$ if and only if $a^{-1}b\in H$.
    \end{enumerate}
\end{lemma}

We now prove.
\begin{lemma}\label{bHbellH}
    Let $H\le G$ and $b\in G$. Let $\ell$ be the minimal integer such that $b^\ell\in H$. Then 
    \begin{enumerate}
        \item\label{BBB1} $\langle b,H\rangle=\langle b\rangle H=\{e,b,\ldots,b^{\ell-1}\}H$.
        \item\label{BBB2} The cosets $H,bH,\ldots,b^\ell H$ are pairwise disjoint.
        \item\label{BBB3} $\ell$ divides $|G|/|H|$.
        \item\label{BBB4} If $b^w=e$, then $\ell|w$.
    \end{enumerate}
\end{lemma}

\begin{proof}
    $\langle b,H\rangle=\langle b\rangle H$ follows from the fact that $G$ is Abelian. Let $b^\ell=h\in H$. Then $b^iH=b^{(i\mod \ell)}h^{\lfloor i/\ell\rfloor}H=b^{(i\mod \ell)}H$.  Therefore, $\langle b\rangle H=\{e,b,\ldots,b^{\ell-1}\}H$. This proves item~\ref{BBB1}. 
    
    If, to the contrary, there are $\ell>i> j\ge 0$ such that $b^iH=b^jH$, then $b^{i-j}\in H$ and, since $0<i-j<\ell$, a contradiction to the minimality of $\ell$. Therefore, $b^{i-j}\not\in H$ and $b^iH\cap b^jH=\emptyset$. This proves item~\ref{BBB2}. 

    Since $\langle b,H\rangle\le G$ we have $|\langle b,H\rangle|=\ell |H|$ which divides $|G|$, and therefore, $\ell$ divides $|G|/|H|$. This proves item~\ref{BBB3}.

    If $b^w=e$, then $H=b^wH=b^{w\mod \ell}H$, and therefore $w\mod \ell=0$. This proves item~\ref{BBB4}.
\end{proof}

\subsection{The Monomial Abelian Group}
Let $n\in\mathbb{N}$, $n\ge 2$, and $K=(\kappa_1,\ldots,\kappa_t)\in \mathbb{N}^t$ such that $\kappa_i\ge 2$ and $\kappa_1\kappa_2\cdots \kappa_t=n$. Let $0\le \ell_{i,j}\le \kappa_i-1$ for $i=2,\ldots,t$ and $j=1,\ldots,i-1$ and let $L=(\ell_{i,j})$. We define 
   $$\Gamma(K,L)=\langle x_1,x_2,\ldots,x_t\ |\ x_i^{\kappa_i}= x_1^{\ell_{i,1}}\cdots x_{i-1}^{\ell_{i,i-1}},i=2,\ldots,t\ ;\ x_1^{\kappa_1}=1\rangle$$
the set of all monomials over the variables\footnote{Monomials over $\mathbb{R}$, so in particular, $x_ix_j=x_jx_i$.} $x_1,\ldots,x_t$ of the form $x_1^{j_1}x_2^{j_2}\cdots x_t^{j_t}$, where $0\le j_i\le \kappa_i-1$ for all $i\in [t]$, with multiplication modulo $x_i^{\kappa_i}- x_1^{\ell_{i,1}}\cdots x_{i-1}^{\ell_{i,i-1}}$ for all $i=2,\ldots,t$ and $ x_1^{\kappa_1}-1$. 
In this group, $\{x_1,\ldots,x_t\}$ is a set of generators of the group with the relations $x_i^{\kappa_i}= x_1^{\ell_{i,1}}\cdots x_{i-1}^{\ell_{i,i-1}}$ for all $i=2,\ldots,t$. Such relations are called {\it triangular relations}. 

This group is a subset of the quotient ring (or factor ring) $${\mathbb{R}}[x_1,x_2,\ldots,x_t]/(x_i^{\kappa_i}- x_1^{\ell_{i,1}}\cdots x_{i-1}^{\ell_{i,i-1}}, i=2,\ldots,t\ ;\ x_1^{\kappa_1}-1)$$ of the polynomial ring modulo the ideal generated by $x_i^{\kappa_i}- x_1^{\ell_{i,1}}\cdots x_{i-1}^{\ell_{i,i-1}},i=2,\ldots,t$ and $x_1^{\kappa_1}-1$.

\noindent
{\bf Example:} Let $n=36$, $K=(4,3,3)$, and 
$L=(\ell_{2,1},\ell_{3,1},\ell_{3,2})=(3,2,1).$ Then
$$\Gamma(K,L)=\langle x_1,x_2,x_3\ |\ x_1^4=1;\ x_2^3=x_1^3;\ x_3^3=x_1^2x_2\rangle.$$
The multiplication of $x_1^2x_2^2x_3^2$ with $x_1^3x_2x_3^2$ is
$$(x_1^2x_2^2x_3^2)(x_1^3x_2x_3^2)=x_1^5x_2^3x_3^4=x_1x_1^3(x_1^2x_2x_3)=x_1^2x_2x_3.$$

The following lemma shows that for any $K$ and $L$ with the above constraints, the set $\Gamma(K,L)$ is an Abelian group. 
\begin{lemma}\label{GamisAb}
   Let $\kappa_1,\ldots,\kappa_t\in \mathbb{N}$ such that $\kappa_i\ge 2$ and $\kappa_1\kappa_2\cdots \kappa_t=n$. Let $0\le \ell_{i,j}\le \kappa_j-1$ for $i=2,\ldots,t$ and $j=1,\ldots,i-1$. Then, for $K=(\kappa_i)$ and $L=(\ell_{i,j})$,
   $\Gamma(K,L)$
   is an Abelian group of order $|\Gamma(K,L)|=n$.
\end{lemma}
\begin{proof}
    First, it is clear that $\Gamma$ is closed under product, and the elements of $\Gamma$ are $x_1^{r_1}\cdots x_t^{r_t}$ for $0\le r_i\le \kappa_i-1$, and therefore $|\Gamma|=\kappa_1\kappa_2\cdots \kappa_t=n$. Since $\Gamma$ is a subset of the quotient ring 
    $${\mathbb{R}}[x_1,x_2,\ldots,x_t]/(x_i^{\kappa_i}- x_1^{\ell_{i,1}}\cdots x_{i-1}^{\ell_{i,i-1}},i=2,\ldots,t\ ;\ x_1^{\kappa_1}-1)$$ and is closed under product, $\Gamma$ is commutative and associative with the unit element $x_1^0x_2^0\cdots x_t^0=1$. Therefore, $\Gamma$ is an Abelian monoid\footnote{Abelian Group but without the requirement for each element to have an inverse.}. It remains to show that every element in $\Gamma$ has an inverse. 

    It is enough to show that every $x_i$ has an inverse. We prove by induction that $x_i^{\kappa_1\kappa_2\cdots \kappa_i}=1$, and therefore $x_i^{\kappa_1\kappa_2\cdots \kappa_i-1}$ is the inverse of $x_i$. First, by the definition of $\Gamma$, we have $x_1^{\kappa_1}=1$. Suppose this is true for $x_1,\ldots, x_{i-1}$. Then
    \begin{eqnarray*}
       x_i^{\kappa_1\kappa_2\cdots \kappa_i}&=&(x_i^{\kappa_i})^{\kappa_1\kappa_2\cdots \kappa_{i-1}}=\left(x_1^{\ell_{i,1}}\cdots x_{i-1}^{\ell_{i,i-1}}\right)^{\kappa_1\kappa_2\cdots \kappa_{i-1}}\\
    &=&(x_1^{\kappa_1})^{\ell_{i,1}\kappa_2\cdots \kappa_{i-1}}(x_2^{\kappa_1\kappa_2})^{\ell_{i,2}\kappa_3\cdots \kappa_{i-1}}\cdots 
    (x_{i-1}^{\kappa_1\kappa_2\cdots \kappa_{i-1}})^{\ell_{i,i-1}}=1. 
    \end{eqnarray*}
\end{proof}

We will call $\Gamma(K,L)$ the {\it monomial group generated by $K$ and $L$}.

We now prove that $\Gamma=\Gamma(K,L)$ is isomorphic to a group $G$ with a set of generators that satisfies similar relations. 
\begin{lemma}\label{IsoGamG}
    Let $G$ be an Abelian group and $A=\{a_1,\ldots,a_t\}$ be a set of generators that satisfy the triangular relations:
    $$a_1^{k_i}=1;\ \  a_i^{k_i}=a_1^{\lambda_{i,1}}\cdots a_{i-1}^{\lambda_{i,i-1}}$$
    for $i=2,\ldots,t$ where $k_i\ge 2$ is the smallest integer such that, for all $i\in [t]$,
    $$a_i^{k_i}\in\langle a_1,\ldots,a_{i-1}\rangle.$$
    Then, $G$ is isomorphic to $\Gamma(K,L)$, where $K=(k_{i})$ and $L=(\lambda_{i,j})$, with the isomorphism
    $$\Psi(x_1^{j_1}\cdots x_t^{j_t})=a_1^{j_1}\cdots a_t^{j_t}$$ for every $0\le j_i\le k_i-1$ and $i\in [t]$.
\end{lemma}
\begin{proof}
    We first prove that $\Psi$ is a homomorphism and then show that it is injective and subjective. 
    
    {\bf Homomorphism}: By Lemma~\ref{PHIPHI} in the appendix, we show that $\Psi(x_1^{\ell_1}\cdots x_t^{\ell_t})=a^{\ell_1}\cdots a^{\ell_t}$ for all $\ell_i\in \mathbb{N}$, $i\in [t]$ and therefore (Here and throughout this proof $0\le j_i,j_i'\le k_i-1$)
    $$\Psi(x_1^{j_1}\cdots x_t^{j_t}\cdot x_1^{j'_1}\cdots x_t^{j'_t})=\Psi(x_1^{j_1+j_2}\cdots x_t^{j_t+j_t})=a_1^{j_1+j_2}\cdots a_t^{j_t+j_t}=\Psi(x_1^{j_1}\cdots x_t^{j_t})\cdot \Psi(x_1^{j'_1}\cdots x_t^{j'_t}).$$
    
    {\bf Injectivity}: If $\Psi(x_1^{j_1}\cdots x_t^{j_t})=\Psi(x_1^{j'_1}\cdots x_t^{j'_t})$, then $a_1^{j_1}\cdots a_t^{j_t}=a_1^{j'_1}\cdots a_t^{j'_t}$. Suppose, to the contrary that $(j_1,\ldots,j_t)$ $\not=(j_1',\ldots,j_t')$, and let $\ell$ be such that $(j_t,j_{t-1},\ldots,j_{\ell+1})=(j'_t,j'_{t-1},\ldots,j'_{\ell+1})$ but $j_\ell\not= j_\ell'$. Assuming without loss of generality that $j_\ell>j_\ell'$, we have
    $$a_\ell^{j_\ell-j_\ell'}=a_1^{j_1'-j_1}\cdots a_{\ell-1}^{j_{\ell-1}'-j_{\ell-1}}\in \langle a_1,\cdots,a_{\ell-1}\rangle.$$ Since $0<j_\ell-j_\ell'<k_\ell$, this leads to a contradiction.  Therefore, $(j_1,\ldots,j_t)=(j_1',\ldots,j_t')$, and hence $x_1^{j_1}\cdots x_t^{j_t}=x_1^{j'_1}\cdots x_t^{j'_t}$.

    {\bf Surjectivity}: Since $A$ is a set of generators, every element $a$ in $G$ can be represented as $a_1^{\alpha_1}\cdots a_t^{\alpha_t}$. By the proof in the appendix, the element $x=x_1^{\alpha_1}\cdots x_t^{\alpha_t}\in \Gamma(K,L)$ satisfies~$\Psi(x)=a$.
\end{proof}

The next lemma shows that multiplying two elements in $\Gamma$ and finding an inverse can be done in poly$(\log |\Gamma|)$ time.
\begin{lemma} \label{MulInv} 
    The multiplication of two elements and finding the inverse of an element in $\Gamma(K,L)$, where $K=(\kappa_i)$ and $L=(\ell_{i,j})$, can be computed in time $\tilde O(\log^2 n)$ where $n=|\Gamma(K,L)|$. 
\end{lemma}
\begin{proof}
    To multiply two elements $\alpha=x_1^{r_1}\cdots x_t^{r_t}$ with $\beta=x_1^{s_1}\cdots x_t^{s_t}$, where $0\le r_i,s_i\le \kappa_i-1$ for all $i\in [t]$, we proceed as follows. First, $\alpha\beta=x_1^{r_1+s_1}\cdots x_t^{r_t+s_t}$. Note that $r_i+s_i< 2\kappa_i$ may be greater than $\kappa_i$. To find the corresponding element in $\Gamma(K,L)$, we use the identity
    $$x_i^m=x_i^{(m\mod \kappa_i)}x_i^{\kappa_i\lfloor m/\kappa_i\rfloor}=x_i^{(m\mod \kappa_i)}x_{i-1}^{\lfloor m/\kappa_i\rfloor\ell_{i,i-1}}\cdots x_{1}^{\lfloor m/\kappa_i\rfloor\ell_{i,1}}$$
    and apply it for $i=t$, then $i=t-1$, down to $i=1$. 

    So, first $x_t^{r_t}x_t^{s_t}=x_t^{r_t+s_t}$. If $r_t+s_t\ge \kappa_t$ then we get $x_t^{r_t+s_t}=x_1^{\ell_{t,1}}\cdots x_{t-1}^{\ell_{t,t-1}} x_t^{r_t+s_t-\kappa_t}$. Therefore, $\alpha\beta=x_1^{r_1+s_1+\ell_{t,1}}\cdots x_{t-1}^{r_{t-1}+s_{t-1}+\ell_{t,t-1}}x_t^{r_t+s_t-\kappa_t}$. Now, for $i\le t-1$, $r_i+s_i+\ell_{t,i}< 3\kappa_i$. We then apply the same procedure for $i=t-1$ and so on.

    It is easy to prove by induction that, when we reach $x_{t-j}$, the exponent of each $x_i$ for $i\le t-j$ is less than $(2^{j+2}+1)\kappa_i$, and the exponent of each $x_i$ for $i>t-j$ is between $0$ and $\kappa_i$. Since $t\le \log n$, we have $(2^{j+2}+1)\kappa_i< 2n^2$. Therefore, all the arithmetic computations are performed with numbers less than $2n^2$. 
    
    Computing $x_i^m$ takes time $O(\log n\log\log n)$,\cite{Mul21}, and updating the exponents of variables $x_{i-1},$ $x_{i-2},\ldots,x_1$ also takes time $O(\log n\log\log n)$. Therefore, the time complexity is at most $\tilde O(\log^2 n)$. 

    To compute the inverse of $\alpha=x_1^{r_1}x_2^{r_2}\cdots x_t^{r_t}$, we first write $\alpha^{-1}=x_1^{-r_1}x_2^{-r_2}\cdots x_t^{-r_t}$. Then, we use the identity
    $$x_i^{-m}=x_i^{\lceil m/\kappa_i\rceil \kappa_i-m}x_{i-1}^{-\lceil m/\kappa_i\rceil\ell_{i,i-1}}\cdots x_1^{-\lceil m/\kappa_i\rceil \ell_{i,1}}.$$
    Following a similar argument to the above, we obtain the result. 
\end{proof}

\subsection{Basis of Abelian Group via Smith Normal Form}

It is well known that every finite Abelian group \( G \) is isomorphic to a direct product
of cyclic groups 
$
G_1 \times G_2 \times \cdots \times G_t,
$
where each \( G_i \) is a cyclic group of order \(m_j\ge 2 \). If \( a_i \) generates the cyclic group \( G_i \), \( i = 1, 2, \ldots, t \), i.e., $G_i=\langle a_i\rangle$,
then the elements \( a_1, a_2, \ldots, a_t \) are called a {\it basis} of \( G \). Since $\langle a_i\rangle$ is isomorphic to $\ZZ_{m_i}$ for $m_i=|\langle a_i\rangle|$, it follows that $G$ is isomorphic to $\ZZ_{m_1}\times \cdots\times \ZZ_{m_t}$. The {\it rank} of the group is the maximum number of cyclic subgroups of prime power order that generate the group. For example, $\ZZ_{12}\times \ZZ_2^2$ is of rank $4$ since it is isomorphic to $\ZZ_{3}\times \ZZ_4\times \ZZ_2\times \ZZ_2$.

The Smith normal form is a fundamental tool for converting generators with relations into a basis. Here, we demonstrate how it applies to the monomial Abelian group.

Given a monomial Abelian group over the variables $x_1,\ldots,x_t$ with the relations
$$x_i^{\kappa_i}= x_1^{\ell_{i,1}}\cdots x_{i-1}^{\ell_{i,i-1}},i=2,\ldots,t\ ;\ x_1^{\kappa_1}=1.$$
Define the $t\times t$ {\it relation matrix} of $\Gamma$ as
$$R=\left[
\begin{array}{cccccc}
-k_t&\ell_{t,t-1}&\ell_{t,t-2}&\cdots&\ell_{t,2}&\ell_{t,1}\\
0&-k_{t-1}&\ell_{t-1,t-2}&\cdots&\ell_{t-1,2}&\ell_{t-1,1}\\
0&0&-k_{t-2}&\cdots&\ell_{t-2,2}&\ell_{t-2,1}\\
\vdots&\vdots&\ddots&\ddots&\vdots&\vdots \\
0&0&0&\cdots&0&-k_1
\end{array}
\right].$$
Smith~\cite{Steph1862} shows that there exist two $t\times t$ invertible (unimodular) matrices $U$ and $V$ with integer entries, and $r\le t$, such that
$$URV=diag(m_1,m_2,\ldots,m_r,0,\ldots,0)$$
with $m_1|m_2|\cdots|m_r$. The resulting diagonal matrix is called the {\it Smith normal form}. It became evident to the group theory community (and is straightforward to prove)) that the Abelian group $\Gamma$ is isomorphic to $\ZZ_{m_1}\times \ZZ_{m_2}\times \cdots\times\ZZ_{m_r}$; see, for example,~\cite{Cohen96,Newman72}. Additionally, if $V=(v_{i,j})$, then the elements 
$$y_i=\prod_{j=1}^t x_j^{v_{i,t-j}}$$
$i=1,2,\ldots,r$, form a basis for $\Gamma$. 

Kannan and Bachem~\cite{KannanB79} gave the first polynomial-time deterministic algorithm for computing the Smith normal form and the matrix $V$. Here, polynomial time is defined as $poly(t,\log\|R\|)$, where $\|R\|=\max_{i,j}|R_{i,j}|$. This implies the following.
\begin{lemma}\label{ConBasis}
    There is a deterministic poly$(t,\log n)$ time algorithm that, for any monomial Abelian group $\Gamma(K,L)$, returns integers $m_1|m_2|\cdots|m_r$ such that $\Gamma(K,L)$ is isomorphic to $\ZZ_{m_1}\times \ZZ_{m_2}\times \cdots\times\ZZ_{m_r}$. The algorithm also provides a basis $y_1,\ldots,y_r$ for $\Gamma(K,L)$.
\end{lemma}

There is also a randomized algorithm~\cite{KaltofenV05} with better time complexity, but the above lemma is sufficient for our purposes. 

\section{Generators for Abelian Group - Deterministic Algorithm}
In this section, we present a deterministic algorithm that runs in time $\tilde O(|G|)$ and identifies generators with triangular relations for the Abelian group $G$. 

We prove
\begin{theorem}\label{AlgGAbelian}
    Given a Cayley multiplication table of an Abelian group $G$, where each entry in the table can be accessed in constant time. There exists a deterministic algorithm that runs in time $\tilde O(|G|)$ and finds a set of generators $A$ for $G$ of size at most $\log |G|$.
\end{theorem}

\begin{algorithm}[H]
\caption{{\bf Generators}$(G)$}\label{AbelGen}
\begin{algorithmic}[1]
    \STATE $A_0 = \{\};\ G_0=\{e\};\  i = 0;$
    \WHILE{$G \neq G_i$}
        \STATE $i \gets i + 1$
        \STATE Choose $a_i \in G \setminus G_{i-1}$.
        \STATE $A_i=A_{i-1}\cup \{a_i\}$.
        \STATE \label{Findki}Find the minimal $k_i$ such that $a_i^{k_i} \in G_{i-1}.$
        \STATE\label{Unioni} $G_i = G_{i-1} \cup \bigcup_{j=1}^{k_i-1}  a_i^jG_{i-1}$.
    \ENDWHILE
    \STATE Output $t:=i$ and $A:=A_t$.
\end{algorithmic}
\end{algorithm}
\begin{proof}
     Consider the algorithm {\bf Generators} in  Algorithm~\ref{AbelGen}. The following observations hold immediately from Lemma~\ref{bHbellH}: 
    \begin{enumerate}
        \item\label{ABA01} $G_{i-1}< G_i=\langle A_i\rangle$.
        \item\label{ABA02} For every two distinct indices $0\le j_1,j_2\le k_i-1$, we have $a_i^{j_1}G_{i-1}\cap a_i^{j_2}G_{i-1}=\emptyset$.
        \item\label{ABA03} $G_{i-1}\cap \bigcup_{j=1}^{k_i-1} a_i^j G_{i-1}=\emptyset$.
    \end{enumerate} 
    Since $G_{i-1}<G_i$, we have $|G_{i-1}|$ divides $|G_i|$ and thus $|G_i|\ge 2|G_{i-1}|$. Consequently, $|A|\le \log |G|$ and the while loop runs at most $\log|G|$ times. 

    Items 1-3 above show, in particular, that each element in $G$ is generated exactly once in some~$G_i$. After computing $a_i^1,a_i^2,\ldots,a_i^{k_i}$ in step~\ref{Findki}, at least $k_i$ elements of the group are added to $G_i$ in step~\ref{Unioni}. Therefore, the algorithm performed at most $|G|$ set operations\footnote{The set operations performed include checking if an element is in the set and adding an element to the set.}. Since each set operation can be performed in $O(\log |G|)$ time, the algorithm's time complexity is $\tilde O(|G|)$.
\end{proof}

The fact that each element in $G$ is generated exactly once in some $G_i$, together with items~\ref{ABA01}-\ref{ABA02} in the proof above and Lemma~\ref{bHbellH}, leads to the following result.
\begin{lemma}\label{ProAbb} Let $A=A_t=\{a_1,\ldots,a_t\}$. Then the following hold.
    \begin{enumerate}
        \item\label{AbPP1} $G_i=\langle a_1,a_2,\ldots,a_i\rangle$ and $G=G_t=\langle a_1,\ldots,a_t\rangle$.
        \item\label{AbPP2} $a_i\not\in G_{i-1}=\langle a_1,a_2,\ldots,a_{i-1}\rangle$, and $k_i$ is the smallest integer in $\mathbb{N}$ such that $a_i^{k_i}\in \langle a_1,a_2,\ldots,a_{i-1}\rangle$.
        \item\label{exept3} Every element in $G$ has a unique representation of the form $a_1^{j_1}\cdots a_t^{j_t}$ for some $j_i\in \{0,1,\ldots,k_i-1\}$.
        \item\label{ProAbb4} $k_i$ divides $|G|/|G_{i-1}|$.
        \item $|G_i|=k_1k_2\cdots k_i$ and $|G|=k_1k_2\cdots k_t$.
    \end{enumerate}   
\end{lemma}

Since $a_i^{k_i}\in \langle a_1,\ldots,a_{i-1}\rangle$, we have 
\begin{eqnarray} \label{aiski} a_i^{k_i}=a_1^{\lambda_{i,1}}\cdots a_{i-1}^{\lambda_{i,i-1}}
\end{eqnarray} for some integers $\lambda_{i,j}\le k_i-1$. Given that each element in $G$ is generated exactly once in some $G_i$, this representation is unique. By Lemma~\ref{IsoGamG}, the group $G$ is isomorphic to
$$\Gamma(K,L)=\langle x_1,x_2,\ldots,x_t\ |\ x_i^{k_i}= x_1^{\lambda_{i,1}}\cdots x_{i-1}^{\lambda_{i,i-1}},i=2,\ldots,t\ ;\ x_1^{k_1}=1\rangle$$
where $K=(k_i)$ and $L=(\lambda_{i,j})$, with the isomorphism 
$$\Psi(x_1^{j_1}x_2^{j_2}\cdots x_t^{j_t})=a_1^{j_1}a_2^{j_2}\cdots a_t^{j_t}.$$

The following algorithm mirrors Algorithm {\bf Generators} but also returns all $k_i$ and $\lambda_{i,j}$ satisfying~(\ref{aiski}). Additionally, the algorithm provides the (unique) representation for each element $a\in G$. Let this representation be 
$$a=a_1^{\lambda_1(a)}\cdots a_t^{\lambda_t(a)}.$$
We also denote $\lambda(a)=(\lambda_1(a),\ldots,\lambda_t(a))$.
\begin{algorithm}[H]
\caption{{\bf GeneratorPlus}$(G)$}\label{AbelGen2}
\begin{algorithmic}[1]
    \STATE \label{StkS}$A_0 = \{\};\ G_0=\{e\};\  i = 0; (\forall a\in G) (\forall j\le \log |G|)\lambda_j(a)=0$
    \WHILE{$G \neq G_i$}\label{WhileGGi}
        \STATE $i \gets i + 1$
        \STATE\label{ChooseaG} Choose $a_i \in G \setminus G_{i-1}$.
        \STATE $A_i=A_{i-1}\cup \{a_i\}$.
        \STATE \label{Stkk}Find the minimal $k_i$ such that $a_i^{k_i} \in G_{i-1}.$
        \STATE $\lambda_{i,j}=\lambda_j(a_i^{k_i})$ for all $j=1,\ldots,i-1$.
        \STATE \label{Stk}$G_i = G_{i-1} \cup \bigcup_{j=1}^{k_i-1}  a_i^jG_{i-1}$.
        \STATE For every $j=1,\ldots,k_i-1$ and $a\in G_{i-1}$, $\lambda(a_i^ja)=\lambda(a), \lambda_i(a_i^ja)=j$;
    \ENDWHILE
    \STATE Output $t:=i$, $A:=A_i$, $(k_1,k_2,\ldots,k_t)$ and a Table for $\lambda(a), a\in G$.
\end{algorithmic}
\end{algorithm}

Since $t\le \log |G|$ we get
\begin{theorem}\label{AlgGAbelian2}
    Given a Cayley multiplication table of an Abelian group $G$, where each entry in the table can be accessed in constant time, there exists a deterministic algorithm that accesses the table $O(|G|)$ times and runs in time $O(|G|\log |G|)$. This algorithm provides 
    \begin{enumerate}
        \item A set $A=\{a_1,\ldots,a_t\}$ of generators for $G$ of size at most $\log |G|$.
        \item Integers $k_i\ge 2$ and $0\le \lambda_{i,j}\le k_j-1$ that satisfies 
        \begin{eqnarray*}  a_i^{k_i}=a_1^{\lambda_{i,1}}\cdots a_{i-1}^{\lambda_{i,i-1}}.
        \end{eqnarray*}
        \item A table $\lambda(a)=(\lambda_1(a),\ldots,\lambda_t(a))\in \prod_{i=1}^{t}\{0,1,\ldots,k_{i}-1\}$ that, for every $a\in G$, provides the unique representation
        $$a=a_1^{\lambda_1(a)}a_2^{\lambda_2(a)}\cdots a_t^{\lambda_t(a)}.$$
    \end{enumerate}
\end{theorem}
\color{black}

\color{black}
\section{Generators for Abelian Group - Randomized Algorithm}
In this section, we prove
\begin{theorem}\label{AlgGAbelianR}
    Given a Cayley multiplication table of an Abelian group $G$, where each entry in the table can be accessed in constant time. There is a randomized algorithm that runs in time $\tilde O(\sqrt{|G|})$ and, with probability $1-1/poly(|G|)$, finds 
    \begin{enumerate}
        \item A set of generators $A=\{a_1,\ldots,a_t\}$ for $G$ of size at most $t\le \log |G|$.
        \item Integers $k_i$ and $0\le \lambda_{i,j}\le k_i-1$ that satisfy the relation 
        \begin{eqnarray*}  a_i^{k_i}=a_1^{\lambda_{i,1}}\cdots a_{i-1}^{\lambda_{i,i-1}}.
        \end{eqnarray*}
    \end{enumerate}
    Additionally, the algorithm generates subsets $A_i\subset G$ and subgroups $G_i$ of $G$ that satisfy all the items stated in Lemma~\ref{ProAbb}.
\end{theorem}
We will first present the proof for the FS-model, and then for the PS-model. Recall that in the FS-model, the size of the group is known.

\subsection{The algorithm in the FS-model}
Consider the following algorithm.
\begin{algorithm}[H]
\caption{{\bf Random Generators}$(G)$}\label{AbelGenR}
\begin{algorithmic}[1]
    \STATE \label{StkSR}$A_0 = \{\};\ G_0=\{e\}; i=0; k_0=1$.
    \WHILE{$k_0k_1k_2\cdots k_i \not= |G|$}
        \STATE $i \gets i + 1$
        \STATE\label{rrr1} $a_i\gets${\bf Choose}$(G,G_{i-1})$\ \hspace{.7in}\texttt{/*} Choose $a_i \in G \setminus G_{i-1}$ where $G_{i-1}:=\langle a_1,a_2,\ldots,a_{i-1}\rangle.$
        \STATE $A_i=A_{i-1}\cup \{a_i\}$.
        \STATE \label{StkkR}$k_i\gets${\bf FindMin}$(a_i,G_{i-1})$\hspace{.63in}\texttt{/*}Find the minimal $k_i$ such that $a_i^{k_i} \in G_{i-1}$
        \STATE\label{RFlam} $(\lambda_{i,j})\gets$ {\bf FindExp}$(a_i^{k_i},G_{i-1})$ \hspace{.25in}\texttt{/*}Find $0\le \lambda_{i,j}\le k_i-1$ such that $a_i^{k_i}=a_1^{\lambda_{i,1}}\cdots a_{i-1}^{\lambda_{i,i-1}}$.
    \ENDWHILE
    \STATE Output $t:=i$, $A:=A_i$, $(k_1,k_2,\ldots,k_t)$ and $\lambda_{m,j}$.
\end{algorithmic}
\end{algorithm}

In the randomized algorithm, {\bf Random Generators}, at each iteration $r$, we do not store the elements of $G_i$ for $i\in [r]$ or their unique representation. For all $i\in [t]$, we store only the integer $k_i$, the set of generators $A_i=\{a_1,\ldots,a_i\}$ for $G_i$, and the relation $a_i^{k_i}=a_1^{\lambda_{i,1}}\cdots a_{i-1}^{\lambda_{i,i-1}}$. This is the same algorithm as {\bf Generators} in Algorithm~\ref{AbelGen2}.

We will show that each one of the procedures {\bf Choose}, {\bf FindMin} and {\bf FindExp} can be executed in time $\tilde O(\sqrt{|G|})$ with a success probability of at least $1-1/poly(|G|)$. The result follows by union bound and since the {\bf while} loop runs at most $O(\log |G|)$ time.

We now prove

\begin{lemma}\label{REl}
    Given $k_1,\ldots,k_r$ and $a_1,\ldots,a_r$, a uniformly random element in $G_r=\langle a_1,a_2,\ldots,a_r\rangle$ can be generated in time $O(\log |G|)$.
\end{lemma}
\begin{proof} Every element in $G_r$ has a unique representation $a_1^{j_1}\cdots a_r^{j_r}$, where $0\le j_i\le k_i-1$. To choose an element uniformly at random, we first choose $0\le j_i\le k_i-1$ uniformly at random for each $i\in [r]$, and then compute $a_1^{j_1}\cdots a_r^{j_r}$ in time
$$r+\lceil \log j_1\rceil + \cdots+\lceil \log j_r\rceil \le 2r+\log(k_1\cdots k_i)\le 3\log |G|.$$
\end{proof}

Notice that we replaced the condition $G\neq G_i$ in step~\ref{WhileGGi} of algorithm {\bf Generators} with the equivalent condition $k_0k_1k_2\cdots k_i\not=|G|$, where $k_0=1$. In the next section, we will discuss the case when $|G|$ is unknown.

We now show how to choose an element $a\in G\backslash G_{i-1}$ in step~\ref{rrr1} in time~$\tilde O(\sqrt{|G|})$.
We prove the following.

\begin{lemma}\label{ainGminGi}
     Let $H<G$ be a subgroup. Suppose we can choose an element uniformly at random from both $H$ and $G$, and multiply two elements in $G$ in constant time. Then, there exists an algorithm that runs in time $\tilde O(\sqrt{|G|})$ and, with probability at least $1-1/poly(|G|)$, 
    \begin{enumerate}
        \item Decides, given $a\in G$, whether $a\in H$.
        \item Finds an element $a\in G\backslash H$.
    \end{enumerate}
\end{lemma}
\begin{proof}
    Let $G$ be a group with $n$ elements. Given $a\in G$. If $a\in H$, then $aH=H$. If $a\not\in H$, then $aH\cap H=\emptyset$.  
    To test if $aH\cap H=\emptyset$, we choose any $t= O(\sqrt{n\log n})$ elements $R_1$ in $aH$ and $t$ elements $R_2$ in $H$ uniformly at random. If $R_1\cap R_2=\emptyset$, we conclude that $a\not\in H$; otherwise, we conclude $a\in H$. If $aH=H$, the probability that the algorithm fails is at most
    $$\left(1-\frac{O(\sqrt{n\log n})}{|H|}\right)^{O(\sqrt{n\log n})}\le \left(1-\frac{O(\sqrt{\log n})}{\sqrt{n}}\right)^{O(\sqrt{n\log n})}=\frac{1}{poly(n)}.$$

    We now prove item 2. Since $H< G$, we have $|H|$ divides $|G|$, so if we choose an element $a\in G$ uniformly at random, with probability at least $1/2$, $a\in G\backslash H$. Thus, the algorithm can choose $O(\log n)$ elements uniformly at random from $G$ and run the algorithm from item~1 for each one. With probability at least $1-1/poly(n)$, one of these will lie in $G\backslash H$. 
\end{proof}

The next lemma shows how to find the smallest $k_i>1$ such that $a_i^{k_i}\in G_{i-1}$ in step~\ref{StkkR} in time $\tilde O(\sqrt{|G|})$. 
\begin{lemma}\label{MinkR} Let $H<G$ be a subgroup of $G$. Suppose we can choose an element uniformly at random from both $H$ and $G$, and can multiply two elements in $G$ in constant time. Given $a\in G$, the smallest $k$ such that $a^k\in H$ can be found in time $\tilde O(\sqrt{|G|})$. 
\end{lemma}
\begin{proof}
   To find the smallest $k$ such that $a^{k}\in H$, we use the fact from Lemma~\ref{bHbellH} that $k$ divides $|G|/|H|$ and therefore divides $n=|G|$. Also, if $k|k'$, then $a^{k'}\in H$. 
   
   The algorithm to find $k$ is as follows: 
   Let $m=n$. For every prime $p$ that divides $m$, while $a^{m/p}$ is in $H$, set $m\gets m/p$. 
   
   The final $m$ is the required $k$. The number of iterations in this algorithm is at most $O(\log n)$. The time complexity is $\tilde O(\sqrt{n})$, as checking if $a^{m/p}$ is in $H$, uses Lemma~\ref{ainGminGi}. 
\end{proof}

Now we show how to find $\lambda_{i,j}$ in step~\ref{RFlam} in time $\tilde O(\sqrt{|G|})$. 
\begin{lemma}\label{Rkilami}
    Let $A=\{a_1,\ldots,a_t\}$ be a set of generators of $G$. The elements $0\le \lambda_{i,j}\le k_i-1$ that satisfy 
        \begin{eqnarray*}  a_i^{k_i}=a_1^{\lambda_{i,1}}\cdots a_{i-1}^{\lambda_{i,i-1}}
        \end{eqnarray*}
    can be found in $\tilde O(\sqrt{n})$ time. 
\end{lemma}
\begin{proof}
    For $a_1$, we have $a_1^{k_1}=1$. Suppose we know the representations of $a_i^{k_i}=a_1^{\lambda_{i,1}}\cdots a_{i-1}^{\lambda_{i,i-1}}$ for $i=1,\ldots,j-1$. Now we have $a_j^{k_j}\in G_{j-1}= \langle a_1,a_2,\ldots,a_{j-1}\rangle$. 
    To find the representation of $a_j^{k_j}$, we choose any $O(\sqrt{n\log n})$ distinct elements in $G_{j-1}$ with known representations. This is achieved by choosing $O(\sqrt{n\log n})$ distinct elements $(w_1,\ldots,w_{j-1})\in \prod_{h=1}^{j-1}\{0,1,\ldots,k_h-1\}$ and computing $a_1^{w_1}\cdots a_{j-1}^{w_{j-1}}$. 
    Using Lemma~\ref{REl}, we choose $O(\sqrt{n\log n})$ elements uniformly at random with known representation and multiply each by $a_j^{k_j}$. If a common element appears, we obtain $(a_1^{\beta_1}\cdots a_{j-1}^{\beta_{j-1}})a_j^{k_j}=a_1^{\alpha_1}\cdots a_{j-1}^{\alpha_{j-1}}$. 
    Thus, we find that $a_j^{k_j}=(a_1^{\alpha_1}\cdots a_{j-1}^{\alpha_{j-1}})
    (a_1^{\beta_1}\cdots a_{j-1}^{\beta_{j-1}})^{-1}$, 
    which, by Lemma~\ref{MulInv}, can be computed in time $\tilde O(\log^2n)$. 

    As in the proof of Lemma~\ref{ainGminGi}, the algorithm succeeds with probability at least $1-1/poly(n)$.
\end{proof}

\subsection{The algorithm in the PS-model}
In this section, we provide a sketch of how to adapt the algorithm {\bf Random Generators} to work within the PS-model. 

Recall that in this model, the size of the group is unknown; the algorithm can receive random, uniformly distributed elements from the group and access the Cayley table of elements that have been observed so far.


We first show how to estimate $|G|$ within a $poly(\log|G|)$ factor. We provide a sketch of the proof.
\begin{lemma}\label{EstimateSG}
     There exists a randomized algorithm that runs in expected time $\tilde O(\sqrt{|G|})$ and, with probability at least $1-1/poly(|G|)$, returns a value $q$ satisfying $|G|\le q\le \Theta(|G|\log |G|)$. 

    The probability that this algorithm requires more than $\Theta(\sqrt{L|G|}\log^2(L|G|))$ elements of $G$ decreases exponentially with $L$.
\end{lemma}
\begin{proof} 
    Here, we present the algorithm along with a sketch of the proof. The full proof is in the Appendix. 
    
    Let $c$ be a sufficiently large constant. For $c$ iterations, we continue to sample uniform random elements from $G$ until a repeat is observed. Let $n_i$ be the number of elements observed at iteration~$i$. Define $n'=\max n_i$. With probability $1-1/poly(|G|)$, we have $\Theta(\sqrt{|G|\log |G|})\ge n'\ge |G|^{1/4}$. 

    Next, for $c\log n'$ iterations, we continue to sample uniform random elements from $G$ until a repeat is observed. Let $m_i$ be the number of elements observed at iteration $i$. Then define $m=\max_i m_i$ and $q=m^2$.
    It can be shown that with probability at least $1-1/poly(|G|)$, $\sqrt{|G|}\le m\le \Theta(\sqrt{|G|\log |G|})$ and therefore $|G|\le q\le O(|G|\log |G|)$. 
\end{proof}

Lemma~\ref{ainGminGi} still applies here if we replace $n$ with $q$. The running time of the algorithm is then $\tilde O(\sqrt{q})=\tilde O(\sqrt{|G|})$.
  
The next lemma shows how to find the smallest $k_i>1$ such that $a_i^{k_i}\in G_{i-1}$ in time $\tilde O(\sqrt{|G|})$. This replaces Lemma~\ref{MinkR}.

\begin{lemma}\label{MinkRPS} Let $G$ be an Abelian group with an unknown number of elements and let $H<G$ be a subgroup of $G$. Suppose we can choose an element uniformly at random from both $H$ and $G$, and can multiply two elements in $G$ in constant time. There is a randomized algorithm that, given $a\in G$, runs in time $\tilde O(\sqrt{|G|})$ and, with probability $1-1/poly|G|$, finds the smallest $k$ such that $a^k\in H$. 
\end{lemma}
\begin{proof}
    The algorithm samples $m=O(\sqrt{q\log q})=\tilde O(\sqrt{|G|})$ integers $j_1,\ldots,j_m$ in $[2q]$ uniformly at random and computes $a^{j_i}$ for each $i\in [m]$. By the birthday paradox (see Appendix~\ref{apB}), with probability at least $1-1/poly(q)=1-1/poly(|G|)$, there are two distinct indices $j_{i_1}< j_{i_2}$ such that $a^{j_{i_1}}=a^{j_{i_2}}$. Consequently, $a^{j_{i_2}-j_{i_1}}=e$. Next, we apply the algorithm from the proof of Lemma~\ref{MinkR}, using $m=j_{i_2}-j_{i_1}$ to find $k$. 
\end{proof}

Finally, Lemma~\ref{Rkilami} also applies in the PS-model.
\color{black}
\color{black}

\section{Isomorphism and Basis for Abelian Groups}
In this section, we prove the following.
\begin{theorem}\label{FindBasis}
    There is a randomized algorithm that accesses the Cayley table of an Abelian group $G$ at most $\tilde O(\sqrt{|G|})$ times and finds a basis for $G$.
\end{theorem}
\begin{proof}
    Given a group $G$, we start by running the procedure {\bf Random Generators} for $G$. By Theorem~\ref{AlgGAbelianR}, this algorithm runs in time $\tilde O(\sqrt{{|G|}})$ and gives a set of generators $A_G$ for $G$ with triangular relations.
    
    Next, we apply Lemma~\ref{IsoGamG} to construct the monomial Abelian group $\Gamma_G$ that is isomorphic to~$G$. 

    By Lemma~\ref{MulInv}, both the multiplication of two elements and computing the inverses in $\Gamma_G$ can be performed in time $\tilde O(\log^2 |G|)$. We then use the algorithm by Kannan and Bachem~\cite{KannanB79} to find a basis for $\Gamma_G$. By Lemma~\ref{ConBasis}, this algorithm runs in time $poly(\log |G|)$. Once we have a basis for $\Gamma_G$, we use the isomorphism between $\Gamma_G$ and $G$ in Lemma~\ref{IsoGamG} to get the basis of $G$.
\end{proof}

We now prove the following.
\begin{theorem}
    The Abelian group isomorphism problem can be solved in time $\tilde O(\sqrt{|G|})$.
\end{theorem}
\begin{proof}
    Given two groups $G$ and $H$, by Theorem~\ref{FindBasis}, we can find the basis for both $G$ and $H$ in time\footnote{Here we assume that $|H| = \Theta(|G|)$. Otherwise, we can run the algorithm on both $G$ and $H$ in parallel and halt when the number of bases of one exceeds that of the other.} $\tilde O(\sqrt{|G|})$. Using the Smith normal form, we can find groups $G'$ and $H'$ of the form $\mathbb{Z}_{m_1} \times \cdots \times \mathbb{Z}_{m_t}$ that are isomorphic to $G$ and $H$, respectively. Then, checking whether $G'$ is isomorphic to $H'$ is straightforward.
\end{proof}

\section{Lower Bounds}

In this section, we give all the lower bounds. 

We begin with some preliminary results, followed by two lemmas that yield the lower bounds.

\subsection{Preliminary Results}
Let $H_1={\mathbb{Z}}_{p^2}^m$ and $H_2=\mathbb{Z}_{p^2}^{m-1}\times \mathbb{Z}_p^2$. 

We first prove.
\begin{lemma}\label{PrElpm}
    For uniformly random elements $h_1,\ldots,h_r\in H_2$ and $h_1',\ldots,h_r'\in H_1$, and any $w_1,\ldots,w_r\in {\mathbb{Z}_{p^2}}$ not all zero
    \begin{eqnarray}
    \Pr[w_1h_1'+\cdots+w_rh_r'=0^{m}]\le  \Pr[w_1h_1+\cdots+w_rh_r=0^{m+1}]\le \frac{1}{p^{m-1}}.\label{NPEq1}
    \end{eqnarray}

    For uniformly random distinct elements $h_1,\ldots,h_r\in H_2$ and $h_1',\ldots,h_r'\in H_1$, and for any $w_1,\ldots,w_r\in {\mathbb{Z}_{p^2}}$ not all zero
    \begin{eqnarray}\label{NPEq2}
    \Pr[w_1h_1'+\cdots+w_r h_r'=0^{m}]\le \Pr[w_1h_1+\cdots+w_rh_r=0^{m+1}]\le \frac{1}{p^{m-1}}+\frac{r^2}{2p^{2m}}.
    \end{eqnarray}

\end{lemma}

\begin{proof}
    We prove (\ref{NPEq1}). If some $w_j$ is not divisible by $p$, then $W=w_1h_1+\cdots+w_rh_r$ is a uniform random element in $H_2$, and the probability that $W=0$ is $1/|H_2|\le 1/p^{2m}$. If $p$ divides all $w_i$, then $w_i=pw_i'$, $W=pW'$, and $W'$ is a uniform random element in $H_2$. Since $pW'=0$ if and only if $p|W'_i$ for all $i\in [m-1]$, the probability is $p^{m-1}p^2/|H_2|=1/p^{m-1}$. 
    
    Similar reasoning yields the result for $H_1$.

    We prove (\ref{NPEq2}). For random uniform $h_1,\ldots,h_r\in H_i$, $i\in\{1,2\}$, the probability that $h_1,\ldots,h_r$ are not distinct is at most
    $r(r-1)/(2|H_i|)\le r^2/(2p^{2m})$. Combining this with (\ref{NPEq1}) gives the result.

\end{proof}

\subsection{Three Lower Bounds Proofs}
We first prove the following lower bound.

\begin{lemma}\label{IsoKnown}
    Let ${\cal G}$ be the set of all groups\footnote{Here and in all the lemmas below, we assume that the set of elements of the group is a subset of a fixed countable set. We make this assumption because, in classical set theory (ZFC), the class (or category) of all groups isomorphic to a given group does not form a set.} that are isomorphic to either $H_1={\mathbb{Z}}_{p^2}^m$ or $H_2=\mathbb{Z}_{p^2}^{m-1}\times \mathbb{Z}_p^2$. Any algorithm that, with probability at least $2/3$, determines whether $G\in {\cal G}$ is isomorphic to $H_1$ or $H_2$, must access the elements of $G$ and its Cayley table at least $\Omega(p^{m/2-1/2})=\Omega({|G|}^{1/4}/\sqrt{p})$ times. 
\end{lemma}
\begin{proof}
    We denote the group sum in  $H_r$, for $r=1,2$, by $+_r$. Consider a set of abstract elements $\Sigma=\{\sigma_1,\ldots,\sigma_t\}$ with $t=p^{2m}$ and a bijective map $\phi:\Sigma\to H_r$ where $r=1,2$. Define the operation $\sigma_i+_{r,\phi}\sigma_j=\phi^{-1}(\phi(\sigma_i)+_r\phi(\sigma_j))$. Then $(\Sigma,+_{r,\phi})$ is a group isomorphic to $H_r$. Let ${\cal G}'=\{(\Sigma,+_{r,\phi})|\phi\ \mbox{bijective}, r=1,2\}.$ Note that the algorithm does not know $\phi$ or $r$, but can compute $\sigma_i+_{r,\phi}\sigma_j$ by accessing the Cayley table. 
    
    We will use Yao's minimax principle. We will show that
    for any deterministic algorithm~${\cal A}$ that decides whether a given group~$G\in {\cal G}'$ is isomorphic to~$H_1$ or~$H_2$ (by outputting~$1$ or~$2$), if~${\cal A}$ is executed on a uniformly random group~$G  \in {\cal G}'$, then with probability at least~$2/3$,~${\cal A}$ must access elements of~$G$ and entries of its Cayley table at least $\Omega(|G|^{1/4}/\sqrt{p})$ times.

    Let ${\cal A}$ be a deterministic algorithm that accesses the elements of a random uniform group $G\in {\cal G}'$ and its Cayley table at most $$T=p^{(m-1)/2}/20$$ times and, with probability at least $2/3$, outputs $r\in\{1,2\}$ if $G$ is isomorphic to $H_r$. We assume that the commands of the algorithm are labeled with numbers $\{1,2,\ldots\}$ and each command is one of the following types:
    \begin{enumerate}[label=Type \arabic*., leftmargin=*, itemsep=0pt]
        \item $z_i\gets \sigma_{j}$.
        \item $z_i\gets z_j+_{r,\phi} z_k$.
        \item If $z_i=z_j$ Goto line $w$.
        \item Output $z_i$
    \end{enumerate}
    Notice that commands of type 1 and 2 are used at most $T$ times, while the command of type $3$ can be used any number of times. This allows the algorithm to search in a table of already accessed elements without any additional cost. 
 
    Define the free $\mathbb{Z}_{p^2}$-module $L=\{\alpha_1x_1+\cdots+\alpha_{w}x_{w}|\alpha_i\in \mathbb{Z}_{p^2},w\in[|G|]\}$ with coefficients in $\mathbb{Z}_{p^2}$, where $\{x_i\}_{i=1}^\infty$ are formal elements. Consider the same algorithm ${\cal A}$, modified so that each $\sigma_j$ in the algorithm is replaced with the formal element $x_j$, each $z_j$ with the new variable $z_j'$, and $+_{r,\phi}$ with the group sum $+$ of $L$. Denote this modified algorithm as ${\cal A}'$. The elements created by ${\cal A}'$ are thus elements of the group $L$. 
    
    We then execute the algorithm ${\cal A}'$ until it terminates, and outputs $r_0\in\{1,2\}$. Note that, since in ${\cal A}'$ each $\sigma_i$ is replaced with elements in $L$, $+_{r,\phi}$ is replaced with $+$ of $L$, and comparisons in the IF commands involve elements of $L$, the execution proceeds along a single, well-defined path $P$ in the algorithm ${\cal A}'$.

    In this execution (path $P$ in ${\cal A}'$), the algorithm ${\cal A}'$ creates at most $T$ variables $z_i'$ where each variable is a linear combination of at most $T$ formal elements $x_i$. We will assume, without loss of generality, that the variables are
    $z_1',z_2',\ldots,z_\ell'$ with $\ell\le T$ and the formal elements are $x_1,\ldots,x_{T}$. Thus, for every $i\in [\ell]$, we have $z_i'=z_{i,1}'x_1+\cdots+z_{i,T}'x_{T}$, $z_{i,j}'\in \mathbb{Z}_{p^2}$, which is an element in $L$. 
    If we follow the same execution path $P$ in ${\cal A}$, disregarding the If command, we obtain $z_i=z_{i,1}'\sigma_1+_{r,\phi}\cdots+_{r,\phi} z_{i,T}'\sigma_{T}.$
    This holds because, in the groups $L$, $H_1$, and $H_2$, the sum is taken modulo $p^2$.

    Therefore, in all the ``If $z_i=z_j$ Goto line $w$'' commands along path $P$, if $z_i'=z_j'$, then $z_i=z_j$. If $z_i'\not=z_j'$, then by Lemma~\ref{PrElpm}, 
    \begin{eqnarray*}      \Pr_{r,\phi}[z_i\not=z_j]&=&\Pr_{r,\phi}[z_{i,1}'\sigma_1+_{r,\phi}\cdots +_{r,\phi} z_{i,T}'\sigma_{T}\not=z_{j,1}'\sigma_1+_{r,\phi}\cdots +_{r,\phi} z_{j,T}'\sigma_{T}]\\
    &=& \Pr_{r,\phi}\left[\phi^{-1}\left(\sum_{r=1}^{T}(z_{i,r}'-z_{j,r})\phi(\sigma_i)\right)\not=0\right]\ \ \ \mbox{Here the sum is $+_r$}\\
    &\ge& \Pr_{r,\phi}\left[\sum_{r=1}^{T}(z_{i,r}'-z_{j,r})\phi(\sigma_i)\not=0^{m+1}\right]\mbox\ \ \mbox{Here the sum is $+_2$}\\
    &\ge& 1-\left(\frac{1}{p^{m-1}}+\frac{T^2}{2p^{2m}}\right).
    \end{eqnarray*}
    Therefore, the probability that for all $i,j\in [T]$, if $z_i'\not=z_j'$, then $z_i\not=z_j$ is at least
    $$1-\left(\frac{T(T-1)}{2}\left(\frac{1}{p^m}+\frac{T^2}{2p^{2m}}\right)\right)\ge 1-\frac{T^2}{p^{m-1}}-\frac{T^4}{p^{2m}}\ge \frac{99}{100}.$$  
    Thus, with probability, at least $99/100$, algorithm ${\cal A}$ follows the same path of execution $P$ as ${\cal A}'$ and output $r_0$. So, if $G$ is isomorphic to $H_r$, and $r\not=r_0$, the algorithm fails with probability at least $99/100$. Therefore, with probability at least $(1/2)(99/100)>1/3$, the algorithm ${\cal A}$ fails. A contradiction.\qed
    
\end{proof}

\begin{lemma}\label{IsunKnown}
    Let ${\cal G}$ be the set of all groups isomorphic to either $H_1=\mathbb{Z}_{p}^{m-1}$ or $H_2=\mathbb{Z}_{p}^m$. Any algorithm that, for $G\in {\cal G}$ of unknown size, decides with probability at least $2/3$ whether $G$ is isomorphic to $H_1$ or $H_2$ must access an oracle that selects uniformly random elements of $G$ and access the Cayley table of $G$ at least $\Omega(p^{m/2-1/2}/\log p)=\Omega({|G|}^{1/2}/(\sqrt{p}\log p))$ times. 
\end{lemma}
\begin{proof}
    Since $p(\mathbb{Z}_{p^2}^{m-1}\times \mathbb{Z}_p^2)$ is isomorphic to $\mathbb{Z}_p^{m-1}$ and $p\mathbb{Z}_{p^2}^m$ is isomorphic to $\mathbb{Z}_{p}^m$, if there exists an algorithm that can distinguish between $H_1=\mathbb{Z}_{p}^{m-1}$ and $H_2=\mathbb{Z}_{p}^m$ in time $T$, then we can solve the problem in Lemma~\ref{IsoKnown} in time $T\log p$. Since by Lemma~\ref{IsoKnown}, $T\log p=\Omega(p^{m/2-1/2})$, the result follows. 
\end{proof}

\begin{lemma}\label{IsoKnownDet}
    Let ${\cal G}$ be the set of all groups that are isomorphic to either $D_1={\mathbb{Z}}_{p}^m$ or $D_2=\mathbb{Z}_{p}^{m-2}\times \mathbb{Z}_{p^2}$. Any deterministic algorithm that determines whether $G\in {\cal G}$ is isomorphic to $D_1$ or $D_2$, must access the elements of $G$ and its Cayley table at least $\Omega(p^{m-2})=\Omega({|G|}/{p^2})$ times. 
\end{lemma}
\begin{proof}
    An adversary can provide elements from $W=\ZZ_p^{m-2}\times \{(0,0)\}$ for each access to an element $\sigma_i$. The sum of any two elements in $W$ remains in $W$, preventing the algorithm from distinguishing between the Abelian groups $D_1$ and $D_2$ until the $(p^{m-2}+1)$-th element is requested. At this point, the adversary must reveal whether the hidden group is $D_1$ or $D_2$. Thus, the complexity of the algorithm is at least $\Omega(p^{m-2})=\Omega(n)$ for any constant $p$. 
\end{proof}

\subsection{Lower Bounds for Isomorphism and Basis}
In this section, we prove the lower bounds.

We now show
\begin{theorem}\label{lb01}
    In the FS-model, the following problems cannot be solved in time less than $\Omega(n^{1/4})$. 
    \begin{enumerate}
        \item Given two Abelian groups $G$ and $H$ of size $n$, decide if $G$ is isomorphic to $H$.
        \item Given an Abelian group $G$ of size $n$, find a basis for $G$.
        \item Given an Abelian group $G$ of size $n$, find a set of generators for $G$ with relations of size at most $n^{o(1)}$. 
    \end{enumerate}
\end{theorem}
\begin{proof}
    If there exists an algorithm that solves the isomorphism problem in time $T$, then the same algorithm can solve the problem in Lemma~\ref{IsoKnown} for $p=O(1)$ in time $T$. Therefore, $T=\Omega(|G|^{1/4})$. This proves item 1.

    If there exists an algorithm that finds a basis for $G$ in time $T$, then the problem in Lemma~\ref{IsoKnown} can be solved in time $O(T)$ as follows: Given $H$ that is either isomorphic to $H_1={\mathbb{Z}}_{p^2}^m$ or $H_2=\mathbb{Z}_{p^2}^{m-1}\times \mathbb{Z}_p^2$, find the basis $\{a_1,\ldots,a_t\}$ for $H$. If $t=m$, then $H$ is isomorphic to $H_1$; otherwise, $t=m+1$ and $H$ is isomorphic to $H_2$. This proves item 2.

    If there is an algorithm that finds generators for $G$ with relations of size at most $t=n^{o(1)}$, then by the result of Kannan and Bachem~\cite{KannanB79} result, one can find the basis of $G$ in time $poly(t,\log n)=n^{o(1)}$. This proves item 3. 
\end{proof}

\begin{theorem}
    In the PS-model, the following problems cannot be solved in time less than $\Omega(|G|^{1/2})$. 
    \begin{enumerate}
        \item Given two Abelian groups $G$ and $H$, decide if $G$ is isomorphic to $H$.
        \item Given an Abelian group $G$, find a basis for $G$.
        \item Given an Abelian group $G$, find generators for $G$ with relations of size at most $n^{o(1)}$. 
    \end{enumerate}
\end{theorem}
\begin{proof}
    We use the same reductions as in the proof of Theorem~\ref{lb01}, combined with Lemma~\ref{IsunKnown}. 
\end{proof}

For deterministic algorithms we have.
\begin{theorem}
    For deterministic algorithm, the following problems cannot be solved in time less than $\Omega(|G|)$. 
    \begin{enumerate}
        \item Given two Abelian groups $G$ and $H$, decide if $G$ is isomorphic to $H$.
        \item Given an Abelian group $G$, find a basis for $G$.
        \item Given an Abelian group $G$, find generators for $G$ with relations of size at most $n^{o(1)}$. 
    \end{enumerate}
\end{theorem}
\begin{proof}
    We use the same reductions as in the proof of Theorem~\ref{lb01}, combined with Lemma~\ref{IsoKnownDet}.
\end{proof}

\bibliography{TestingRef}

\appendix
\section{Two Technical Proofs}
Here we prove
\begin{lemma}\label{PHIPHI}
    Let $G$ be an Abelian group and $A=\{a_1,\ldots,a_t\}$ be a set of generators that satisfy the relation
    $$a_1^{k_1}=1;\ \  a_i^{k_i}=a_1^{\lambda_{i,1}}\cdots a_{i-1}^{\lambda_{i,i-1}}$$
    for $i=2,\ldots,t$ where $k_i\ge 2$ is the smallest integer such that 
    $$a_i^{k_i}\in\langle a_1,\ldots,a_{i-1}\rangle$$
    for all $i\in [t]$. If
    $$\Psi(x_1^{j_1}\cdots x_t^{j_t})=a^{j_1}_1\cdots a^{j_t}_t$$ for every $0\le j_i\le k_i-1$ and $i\in [t]$ then 
    $$\Psi(x_1^{\ell_1}\cdots x_t^{\ell_t})=a^{\ell_1}_1\cdots a^{\ell_t}_t$$ for all $\ell_i\in\mathbb{N}$, $i\in [t]$.
\end{lemma}
\begin{proof}
    We will show by induction on $m$ that $$\Psi(x_1^{\ell_1}\cdots x_m^{\ell_m})=a^{\ell_1}_1\cdots a^{\ell_m}_m$$ for all $\ell_i\in\mathbb{N}$ for $i\in [m]$. The case $m=1$ is easy to verify. Assume the hypothesis holds for $m-1$. For $m$, let $r=\ell_m\mod k_m$ and $s=\lfloor \ell_m/k_m\rfloor$. Suppose $x_1^{\ell_1+s\lambda_{m,1}}\cdots x_{m-1}^{\ell_{m-1}+s\lambda_{m,m-1}}=x_1^{w_1}\cdots x_{m-1}^{w_{m-1}}$ where $w_i\le k_i$ for all $i$.
    Then
    \begin{eqnarray*}
        \Psi(x_1^{\ell_1}\cdots x_{m-1}^{\ell_{m-1}} x_{m}^{\ell_{m}})&=&\Psi(x_1^{\ell_1+s\lambda_{m,1}}\cdots x_{m-1}^{\ell_{m-1}+s\lambda_{m,m-1}} x_{m}^{r})\\
        &=&\Psi(x_1^{w_1}\cdots x_{m-1}^{w_{m-1}}x_m^r)\\
        &=& a_1^{w_1}\cdots a_{m-1}^{w_{m-1}}a_m^r\\
        &=&\Psi(x_1^{w_1}\cdots x_{m-1}^{w_{m-1}})a_m^r\\
        &=& \Psi(x_1^{\ell_1+s\lambda_{m,1}}\cdots x_{m-1}^{\ell_{m-1}+s\lambda_{m,m-1}} )a_{m}^{r}\\
        &=&a_1^{\ell_1+s\lambda_{m,1}}\cdots a_{m-1}^{\ell_{m-1}+s\lambda_{m,m-1}} a_{m}^{r}\\
        &=& a_1^{\ell_1}\cdots a_{m-1}^{\ell_{m-1}} a_{m}^{\ell_{m}}.
    \end{eqnarray*}
\end{proof}

We now prove Lemma~\ref{EstimateSG}.

\noindent
{\bf Lemma}\ \ref{EstimateSG}.\ {\it
    There exists an algorithm that runs in expected time $\tilde O(\sqrt{|G|})$ and, with probability at least $1-1/poly(|G|)$, returns a value $q$ satisfying $|G|\le q\le \Theta(|G|\log |G|)$. 

    The probability that this algorithm requires more than $\Theta(\sqrt{L|G|}\log^2(L|G|))$ elements of $G$ decreases exponentially with $L$.}
\begin{proof}
    Let $c$ be a sufficiently large constant. For $c$ iterations, we continue to sample uniform random elements from $G$ until a repeat is observed. Let $n_i$ be the number of elements observed at iteration~$i$. Define $n'=\max n_i$. 

    Next, for $c\log n'$ iterations, we continue to sample uniform random elements from $G$ until a repeat is observed. Let $m_i$ be the number of elements observed at iteration $i$. Then define $m=\max_i m_i$. Then output $q=m^2$.

    Let $n=|G|$. The main procedure of the algorithm is as follows: continue sampling uniform random elements from $G$ until a repeat is observed. Let $X$ be a random variable that represents the number of elements received in this procedure before a repeat occurs. For any integer $k$, we have
    \begin{eqnarray}
        \Pr[X>k]&=&\Pr[\mbox{The first $k+1$ elements are distinct}]\nonumber\\
        &=&\left(1-\frac{1}{n}\right)\left(1-\frac{2}{n}\right)\cdots \left(1-\frac{k}{n}\right)\nonumber\\
        &\le& e^{-\left(\frac{1}{n}+\frac{2}{n}+\cdots+\frac{k}{n}\right)}\le e^{-\frac{k^2}{2n}}.\label{LBk}
    \end{eqnarray}
    and
    \begin{eqnarray}
        \Pr[X< k]&=&1-\Pr[X>k-1]\nonumber\\
        &=&1-\left(1-\frac{1}{n}\right)\left(1-\frac{2}{n}\right)\cdots \left(1-\frac{k-1}{n}\right)\nonumber\\
        &\le& \frac{1}{n}+\frac{2}{n}+\cdots+\frac{k-1}{n}\le \frac{k^2}{2n}.\label{UBk}
    \end{eqnarray}
    Now, by (\ref{LBk}), and for any $T$ and a constant $c_1$, we have
    \begin{eqnarray}       \Pr\left[n'>c_1\sqrt{nT}\right]\le \sum_{i=1}^c\Pr[n_i>c_1\sqrt{nT}]\le ce^{-\Theta(T)}=e^{-\Theta(T)}.\label{Hiho}
    \end{eqnarray}
    Therefore,
    $$\Pr\left[n'>\Theta(\sqrt{n\log n})\right]\le  \frac{1}{poly(n)}.$$
    By (\ref{UBk}),
    $$\Pr[n'<n^{1/4}]\le \prod_{i=1}^c\Pr[n_i<n^{1/4}]=\left(\frac{n^{1/2}}{2n}\right)^c=\frac{1}{poly(n)}.$$
    Therefore, with probability at least $1-1/poly(n)$, we have $$\Theta(\sqrt{n\log n})\ge n'\ge n^{1/4}.$$ 

    Now, assuming $\Theta(\sqrt{n\log n})\ge n'\ge n^{1/4}$, by (\ref{LBk}) and (\ref{UBk}),
    $$\Pr[m>\Theta(\sqrt{n\log n})]\le c\log n' e^{-\Theta(\log n)}\le \frac{1}{poly(n)},$$
    and
    $$\Pr[m<n^{1/2}]\le \prod_{i=1}^{c\log n'}\Pr[m_i<n^{1/2}]=\left(\frac{n}{2n}\right)^{c\log n'}=\frac{1}{poly(n)}.$$
    Therefore, with probability at least $1-1/poly(n)$, we have $n\le q=m^2\le \Theta(n\log n)$.

    We now show that the probability that this algorithm requires more than $\Theta(\sqrt{Ln}\log^2(Ln))$ elements of $G$ decreases exponentially with $L$.

    Let
    $$N= \sum_{i=1}^cn_i,\ \ \ \ \mbox{and}\ \ \ M=\sum_{j=1}^{c\log n'}m_i.$$
    Let $\lambda=\Theta(\log^2 (Ln))$. The probability that this algorithm requires more than $\lambda\sqrt{Ln}$ elements of $G$ is
    $$\Pr\left[N+M\ge \lambda\sqrt{Ln}\right]\le \Pr\left[N\ge (\lambda/2)\sqrt{Ln}\right]+\Pr\left[M\ge (\lambda/2)\sqrt{Ln}\right].$$
    Now, by (\ref{LBk}), 
    \begin{eqnarray*}
        \Pr\left[N\ge (\lambda/2)\sqrt{Ln}\right]&\le&c\Pr\left[n_1\ge (\lambda/2c)\sqrt{Ln}\right]\\
        &\le& c e^{-(\lambda/2c)^2L/2}\le (Ln)^{-\Theta(L\log(Ln))}.
    \end{eqnarray*}
    And,
    \begin{eqnarray*}
        \Pr\left[M\ge (\lambda/2)\sqrt{Ln}\right]&\le& \Pr\left[\left. M\ge (\lambda/2)\sqrt{Ln}\ \right |\ n'\le \log(Ln)\sqrt{Ln}\right]+\Pr\left[n'> \log(Ln)\sqrt{Ln}\right].
    \end{eqnarray*}
    By (\ref{Hiho})
    $$\Pr\left[n'\ge (\log Ln)\sqrt{Ln}\right]\le (Ln)^{-\Theta(L\log(Ln))}.$$
Given that $n'\le (\log (Ln))\sqrt{Ln}$, by (\ref{LBk}), we have
\begin{eqnarray*}
    \Pr\left[M\ge (\lambda/2)\sqrt{Ln}\right]&\le& c\log n'\Pr\left[m_1\ge (\lambda/(2c\log n'))\sqrt{Ln}\right]\\
    &\le& 2c\log(Ln)\Pr\left[m_1\ge \Theta\left(\log(Ln)\sqrt{Ln}\right)\right]\\&\le& (Ln)^{-\Theta(L\log(Ln))}.
\end{eqnarray*}
Therefore, the probability that this algorithm requires more than $\Theta(\sqrt{Ln}\log^2(Ln))$ elements of $G$ is at most $$(Ln)^{-\Theta(L\log(Ln))}.$$

\section{Birthday Paradox}\label{apB}
We consider the following standard ``birthday paradox'' setting.
There are \(n\) distinct bins and we perform \(m\) independent draws,
where each draw places a ball uniformly at random into one of the \(n\) bins.
A \emph{collision} occurs if at least two balls fall into the same bin.

Let \(P_{\mathrm{coll}}(n,m)\) denote the probability that a collision occurs
after \(m\) draws. Equivalently,
\[
P_{\mathrm{coll}}(n,m) \;=\; 1 - \frac{n(n-1)\cdots (n-m+1)}{n^m}
\;=\;
1 - \prod_{i=0}^{m-1} \left(1 - \frac{i}{n}\right),
\]
where the product term is the probability that all \(m\) balls fall into
distinct bins.

For \(m \le n/2\), we have the standard bounds
\[
\exp\!\left(
    -\frac{1}{n} \sum_{i=0}^{m-1} i
    -\frac{1}{n^2} \sum_{i=0}^{m-1} i^2
\right)
\;\le\;
\prod_{i=0}^{m-1} \left(1 - \frac{i}{n} \right)
\;\le\;
\exp\!\left(
    -\frac{1}{n} \sum_{i=0}^{m-1} i
\right).
\]
That is,
\[
1 - \exp\!\left( -\frac{m(m-1)}{2n} \right)
\;\le\;
P_{\mathrm{coll}}(n,m)
\;\le\;
1 - \exp\!\left( -\frac{m(m-1)}{2n} - \frac{(m-1)m(2m-1)}{6n^2} \right).
\]

Now set \( m = c\sqrt{n} \) with fixed \( c>0 \). Then
\[
\frac{m(m-1)}{2n}
= \frac{c^2}{2} + O\!\left( \frac{1}{\sqrt{n}} \right),
\qquad
\frac{(m-1)m(2m-1)}{6n^2}
= \frac{c^3}{3\sqrt{n}} + O\!\left( \frac{1}{n} \right).
\]
Plugging these into the bounds gives
\[
1 - \exp\!\left( -\frac{c^2}{2} + O\!\left( \tfrac{1}{\sqrt{n}} \right) \right)
\;\le\;
P_{\mathrm{coll}}(n,c\sqrt{n})
\;\le\;
1 - \exp\!\left( -\frac{c^2}{2} - \frac{c^3}{3\sqrt{n}} + O\!\left( \tfrac{1}{n} \right) \right).
\]

In particular, as \( n \to \infty \),
\[
P_{\mathrm{coll}}(n,c\sqrt{n}) \;\longrightarrow\; 1 - e^{-c^2/2},
\]
with error term \( O\!\left( \frac{1}{\sqrt{n}} \right) \).
\end{proof}

\end{document}